\documentclass[11pt]{article}
\usepackage[margin=1.05in]{geometry}
\usepackage[utf8]{inputenc}
\usepackage[theorems]{textool}
\makeatletter
\renewcommand\@fnsymbol[1]{\ensuremath{\ifcase#1\or *\or \dagger\or
	\ddagger\or \mathsection\or \mathparagraph\else\@ctrerr\fi}}
\makeatother
\usepackage{amsthm}
\usepackage{graphicx}
\usepackage{xcolor}
\usepackage[colorlinks=true,linkcolor=blue,citecolor=blue,pdfusetitle]{hyperref}

\DeclareMathOperator{\SO}{SO}

\newcommand{\Hess}{\mathrm{Hess}}

\title{Geodesic strong convexity does not imply forward invariance
	under gradient flow on $\SO(3)$:\\ a certified counterexample}
\author{Dongming Wang, Wei Ren%
	\thanks{The authors are with the Department of Electrical and
		Computer Engineering, University of California, Riverside, CA
		92521, USA.
		\texttt{dongming.wang@email.ucr.edu},
		\texttt{ren@ece.ucr.edu}.}}
\date{}

\hypersetup{pdftitle={Geodesic strong convexity does not imply forward invariance under gradient flow on SO(3): a certified counterexample}, pdfauthor={Dongming Wang, Wei Ren}, pdfsubject={Riemannian optimization; certified computation}, pdfkeywords={SO(3), geodesic convexity, forward invariance, interval arithmetic}}
\begin{document}
\maketitle

\begin{abstract}
	Let $\ccalC=\overline{\ccalB}_{\rho}(R_c)$ be a geodesic ball of
	radius $\rho<\pi/2$ in $\SO(3)$ with the bi-invariant metric,
	and let $f$ be geodesically strongly convex on $\ccalC$ with an
	interior minimizer.
	It is tempting to expect the gradient flow
	$\dot R=R(-\nabla f)^\wedge$ to keep $\ccalC$ forward
	invariant: the flow is attracted to an interior point, and
	strong convexity appears to leave no room for outward motion.
	We show this expectation is false by an explicit, fully
	certified construction with $\rho=0.3$: a cost, quadratic in
	the principal logarithmic chart with off-diagonal coupling
	$0.7$, whose geodesic Hessian satisfies
	$\Hess f\succeq\mu I_3$ on all of $\ccalC$ with a
	machine-certified modulus $\mu\geq0.172$, rigorous ball arithmetic over exact rational inputs, yet whose descent
	velocity at a boundary point has the exact rational outward
	radial component $21/500$.
	A continuity corollary of the exact rate certifies that the
	flow exits the ball; numerical integration puts the peak
	excursion near $0.3143$ before convergence to the minimizer.
	The mechanism is elementary: strong convexity constrains the
	projection of the gradient onto the minimizer direction, not
	onto the inward radial direction.
	Code reproducing every certified constant and figure
	accompanies the note.
\end{abstract}

\section{Introduction}

Gradient flows of geodesically strongly convex functions on
Riemannian manifolds enjoy strong global properties, namely a unique minimizer, exponential contraction of the objective, and exponential convergence of the trajectory, when the flow is globally defined
and remains within the geodesically convex region on which the
modulus holds~\cite{absil2008optimization,
	boumal2023intromanifolds}.
When such flows are used as building blocks of constrained or
distributed algorithms, for instance in attitude coordination
on $\SO(3)$, where all analysis is confined to a geodesic ball
$\ccalC=\overline{\ccalB}_\rho(R_c)$ inside the injectivity
radius~\cite{sarlette2009consensus, tron2013riemannian,
	lin2017distributed}, a natural structural question arises:

\begin{quote}
	\emph{If $f$ is geodesically strongly convex on $\ccalC$ with
	minimizer in the interior of $\ccalC$, is $\ccalC$ forward
	invariant under the Riemannian gradient flow of $f$?}
\end{quote}

If true, boundary behavior would come for free from convexity, and
operating-region assumptions in manifold optimization and control
could be dropped.
The question is subtler than its Euclidean intuition suggests.
In $\mbR^m$, strong convexity of $f$ makes every sublevel set
convex, and the flow decreases $f$; but a metric \emph{ball} need
not be a sublevel set of $f$, and even in $\mbR^m$ the descent
velocity $-\nabla f$ of a strongly convex quadratic at a sphere
around a different center can point outward.
On a curved space the same phenomenon persists; positive
curvature is not the cause, entering only through the
certification burden of the chart geometry.
The point of this note is to pin the failure down on $\SO(3)$
\emph{quantitatively and certifiably}, in a form usable as a
citable obstruction: every
constant below is either exact rational or enclosed by rigorous ball
arithmetic, and the two purely numerical illustrations, a sampled spectrum and an integrated trajectory, are explicitly labeled as
such.

\paragraph{Contributions.}
(i)~An explicit cost on $\overline{\ccalB}_{0.3}(I_3)\subset\SO(3)$,
quadratic in the principal logarithmic chart, whose geodesic
Hessian is certified to satisfy $\Hess f\succeq 0.172\,I_3$ on the
whole ball (Theorem~\ref{thm:convex}); the certificate combines an
analytic lower bound for the principal term with a ball-arithmetic
bound for the Jacobian-derivative correction, organized around
three power series with positive coefficients so that the supremum
is attained at the endpoint and the certificate reduces to one
rigorously evaluated point (Section~\ref{sec:certified}).
(ii)~An \emph{exact} violation of the inward-pointing descent
condition: at the boundary point $R_b=\exp(0.3\,e_1^\wedge)$ the
outward radial speed of the gradient flow equals $21/500$ exactly,
by a block-structure argument that
bypasses all transcendental quantities
(Theorem~\ref{thm:violation}).
(iii)~A trajectory-level demonstration that the failure is not
infinitesimal: a continuity corollary of the exact rate proves the
gradient flow from $R_b$ exits the ball, and numerical integration
puts the peak excursion at $\max_t d(R(t),R_c)\approx0.3143$
before convergence to the interior minimizer
(Section~\ref{sec:flow}), in contrast with the radial cost $\tfrac12 d^2(\cdot,R^*)$ started at the same point.
(iv)~The type-correct second-order calculus needed along the way:
the body-frame Hessian of a chart-quadratic cost, the directional
derivative of the inverse right Jacobian, and the small-angle
extensions ($\beta'(\theta)/\theta\to1/360$) that make every
formula smooth throughout the principal chart
(Section~\ref{sec:hessian}).
(v)~The positive statement that the obstruction calls for: an
inward-pointing descent boundary inequality, which is automatic
for smooth radial costs $\phi(d(\cdot,R^*))$ with $\phi'\geq0$,
$\phi'(0)=0$, and a smooth even radial extension, is sufficient
to restore \emph{strong invariance} of the ball for the
gradient flow and for the signum-gradient dynamics analyzed in
the companion manuscripts, which develop the resulting
distributed theory.

\section{Setup and main results}\label{sec:setup}

Throughout, $\SO(3)$ carries the bi-invariant metric
$\inner{\xi}{\eta}=\xi^\top\eta$ on body-frame velocities, with
geodesics $t\mapsto R\exp(t\xi^\wedge)$ and distance
$d(R_1,R_2)=\norm{(\log(R_1^\top R_2))^\vee}$ for
$d<\pi$~\cite{boumal2023intromanifolds, hall2015lie}; the sectional
curvature is $1/4$.
Fix the center $R_c=I_3$ and radius $\rho=0.3$, and write
$\ccalC:=\overline{\ccalB}_{0.3}(I_3)$.
On the principal chart
$U:=\{R:\mathrm{tr}(R)>-1\}=\{R:d(R,I_3)<\pi\}\supseteq\ccalC$,
let $\phi(R):=(\log R)^\vee$ and define
\feq{
	f(R) \;:=\; \tfrac12\,\bigl(\phi(R)-\phi^*\bigr)^\top
	H\,\bigl(\phi(R)-\phi^*\bigr),
	\qquad
	H=\begin{bmatrix}
		1 & 0.7 & 0\\ 0.7 & 1 & 0\\ 0 & 0 & 1
	\end{bmatrix},
	\quad
	\phi^*=\begin{pmatrix}0.23\\0.16\\0\end{pmatrix}.
	\label{eq:f}
}
\begin{remark}[Global extension]\label{rem:cutoff}
To make $f$ globally defined (twice continuously differentiable on
all of $\SO(3)$), multiply~\eqref{eq:f} by a smooth cutoff
$\chi:\SO(3)\to[0,1]$ with $\mathrm{supp}\,\chi\Subset U$
(compact containment makes the extension by zero smooth across
$\partial U$) and $\chi\equiv1$ on a neighborhood of
$\overline{\ccalB}_{0.5}(I_3)$.
Every statement below concerns the ball $\ccalC$, where
$\chi\equiv1$, so the cutoff plays no quantitative role and is
suppressed from the notation; the extension vanishes outside $\mathrm{supp}\,\chi$ and is
nonnegative, so $R^*$, at value zero, remains a global minimizer
of it, though not the unique one, every point outside
$\mathrm{supp}\,\chi$ attaining the same value; all uniqueness
claims below concern minimization over $\ccalC$.
\end{remark}
What is true, and all that is used, is this: the spectrum of $H$ is
$\{0.3,\,1,\,1.7\}$, the chart minimizer $\phi^*$ satisfies
$\norm{\phi^*}=\sqrt{0.0785}\approx0.280179<\rho$, and
$R^*:=\exp((\phi^*)^\wedge)$ lies in the \emph{interior} of
$\ccalC$ and is the unique minimizer of $f$ \emph{over} $\ccalC$
(indeed $f\geq0$ on the chart with equality only at $\phi^*$).

\begin{theorem}[Certified strong convexity]\label{thm:convex}
	The cost~\eqref{eq:f} satisfies
	\nfeq{
		\Hess f(R)\;\succeq\;\mu\,I_3
		\quad\text{on }\ccalC,
		\qquad
		\mu \;\geq\; 0.172,
	}
	where $\Hess$ is the geodesic (body-frame) Hessian.
	The bound is machine-certified in ball arithmetic with exact
	rational inputs (Proposition~\ref{prop:interval}); a direct
	numerical cross-check over $4000$ random points of $\ccalC$, uniform in the principal-log chart ball, with exact
	Hessians assembled from the closed-form second derivative
	along geodesics and validated against central differences,
	gives a sampled minimum eigenvalue of $0.254$.
\end{theorem}

\begin{theorem}[Exact boundary violation]\label{thm:violation}
	At the boundary point $R_b:=\exp(0.3\,e_1^\wedge)
	\in\partial\ccalC$, with $\hat e_{c}$ the unit tangent at $R_b$
	pointing toward the center $R_c$,
	\nfeq{
		\hat e_{c}^{\,\top}\,\nabla f(R_b)
		\;=\; \frac{21}{500} \;=\; 0.042 \;>\;0
	}
	\emph{exactly}: the gradient's inward projection is positive,
	equivalently the descent velocity $-\nabla f(R_b)$ points
	strictly outward with radial speed $21/500=0.042$~rad/s; the
	immediate exit of the gradient flow is
	Corollary~\ref{cor:exit}.
\end{theorem}

\begin{proposition}[An open family]\label{prop:family}
	Fix $\rho\in(0,\pi/2)$ and consider the family
	$H_c := I_3 + c\,(e_1e_2^\top+e_2e_1^\top)$, $|c|<1$, with
	target $\phi^*=(a,b,0)$ and boundary point
	$R_b=\exp(\rho\,e_1^\wedge)$.
	Then
	\nfeq{
		\hat e_c^\top\nabla f(R_b) \;=\; a+cb-\rho
		\qquad\text{exactly},
	}
	so the initial radial velocity of the gradient flow at $R_b$
	equals $a+cb-\rho$; whenever $a+cb>\rho$ the descent direction
	points strictly outward at $R_b$ and the trajectory exits the
	ball immediately, while $a^2+b^2<\rho^2$ keeps the minimizer
	strictly inside the chart ball.
	The certified instance $(a,b,c,\rho)=(0.23,0.16,0.7,0.3)$
	satisfies both strictly, since $0.23+0.7\cdot0.16=0.342>0.3$
	and $0.23^2+0.16^2=0.0785<0.09$.
	For the openness, define
\nfeq{
	\tilde\mu(a,b,c,\rho)
	:=1-|c|-\bigl(\rho+\sqrt{a^2+b^2}\bigr)(1+|c|)\,C_1(\rho),
}
so that the argument of Theorem~\ref{thm:convex} gives
$\mathrm{Hess}\,f\succeq\tilde\mu(a,b,c,\rho)\,I$ on
$\ccalB_\rho(I_3)$, the eigenvalues of $H$ being $1\pm|c|$
and $1$.
At the certified instance,
$\rho_0+\sqrt{a_0^2+b_0^2}\leq0.58017852$ and
$\tilde\mu(0.23,0.16,0.7,0.3)
\geq 0.3-0.58017852\cdot1.7\cdot C_1(0.3)>0.1726$; since
$\tilde\mu$ is continuous, it stays positive on a parameter
neighborhood, and the two strict inequalities $a+cb>\rho$ and
$a^2+b^2<\rho^2$ persist by ordinary continuity.
If the global cutoff extension is retained for the family,
choose the neighborhood inside $\rho<0.5$ or a cutoff equal to
one on a neighborhood of $\overline{\ccalB}_{0.5}(I_3)$.
Every member is therefore geodesically strongly convex on its
ball yet violates the inward-pointing condition: the
construction is an open family, not an isolated point.
\end{proposition}
\begin{proof}
	With $\phi_b=\rho e_1$ one has $\phi_b^\wedge e_1=0$, so
	$J_r^{-1}(\phi_b)e_1=e_1$; since $\hat e_c=-e_1$ and
	$\nabla_Rf=J_r^{-\top}\nabla_\phi f$ with
	$\nabla_\phi f(\phi_b)=H_c(\rho e_1-\phi^*)$,
	$\hat e_c^\top\nabla f(R_b)
	=-e_1^\top H_c(\rho e_1-\phi^*)=a+cb-\rho$.
	The radial-speed identity
	$\tfrac{d}{dt}d(R,R_c)=\hat e_c^\top\nabla f$ along
	$\dot R=R(-\nabla f)^\wedge$ (Section~\ref{sec:discussion})
	converts this into the initial exit speed.
\end{proof}

Theorems~\ref{thm:convex} and~\ref{thm:violation} together answer
the opening question in the negative: the example satisfies
the standard convexity package with explicit certified
constants, strong convexity on the ball, an interior minimizer,
and bounded gradients, here
$\norm{\nabla f}\leq\lambda_{\max}(H)\,\bar q\,\kappa(0.3)
\leq 0.991$ on $\ccalC$, the product being certified in
$[0.99001,\,0.99002]$, so this particular analytical product
bound cannot be rounded down to $0.99$, and invariance still
fails.
Section~\ref{sec:flow} quantifies the failure at trajectory level.

\section{Second-order calculus in the principal chart}
\label{sec:hessian}

\subsection{Right Jacobian and its derivative}

For $\phi\in\mbR^3$ with $\theta:=\norm{\phi}<\pi$, the inverse
right Jacobian of $\SO(3)$
is~\cite{sola2018micro, chirikjian2011stochastic}
\feq{
	J_r^{-1}(\phi)
	= I_3+\tfrac12\phi^\wedge+\beta(\theta)\,(\phi^\wedge)^2,
	\qquad
	\beta(\theta)
	= \frac{1}{\theta^2}
	\Bigl(1-\frac{\theta}{2}\cot\frac{\theta}{2}\Bigr),
	\label{eq:Jrinv}
}
and along $\dot R=R\omega^\wedge$ the chart coordinate obeys
$\dot\phi=J_r^{-1}(\phi)\,\omega$.
Three scalar functions control all estimates below; each is given on
$[0,\pi)$ by a power series with \emph{positive} coefficients:
\feq{
	\begin{aligned}
	\kappa(\theta)&:=\frac{\theta}{2\sin(\theta/2)}
	= 1+\frac{u}{6}+\frac{7u^2}{360}+\frac{31u^3}{15120}
	+\frac{127u^4}{604800}+\cdots,\qquad u=(\theta/2)^2,\\
	\beta(\theta)&=\frac{1}{12}+\frac{\theta^2}{720}
	+\frac{\theta^4}{30240}+\frac{\theta^6}{1209600}+\cdots,
	\qquad\quad
	\beta'(\theta)=\frac{\theta}{360}+\frac{\theta^3}{7560}
	+\frac{\theta^5}{201600}+\cdots,
	\end{aligned}
	\label{eq:series}
}
the $\beta$-coefficients being $|B_{2k}|/(2k)!$ (Bernoulli
numbers).
In particular $\beta$ and $\beta'$ extend smoothly through
$\theta=0$ with $\beta(0)=1/12$, $\beta'(0)=0$, and
\feq{
	\lim_{\theta\to0}\frac{\beta'(\theta)}{\theta}
	= \frac{1}{360},
	\label{eq:limit}
}
which is the extension needed to make the derivative formula below
smooth throughout the principal chart; throughout, $\kappa$,
$\beta$, and $\beta'$ are
understood at $\theta=0$ through these series, the analytic extensions of the closed forms.
The singular values of $J_r^{-1}(\phi)$ are
$\{1,\kappa(\theta),\kappa(\theta)\}$ with $\kappa\geq1$; the value
$1$ is attained on the axis $\phi$.

Differentiating~\eqref{eq:Jrinv} in a direction $u\in\mbR^3$ (with
$\theta>0$) gives
\feq{
	D J_r^{-1}(\phi)[u]
	= \tfrac12 u^\wedge
	+ \frac{\beta'(\theta)}{\theta}\,(\phi^\top u)\,
	(\phi^\wedge)^2
	+ \beta(\theta)\bigl(u^\wedge\phi^\wedge
	+\phi^\wedge u^\wedge\bigr),
	\label{eq:DJ}
}
where the middle term extends continuously by $0$ at $\phi=0$
via~\eqref{eq:limit} (it is $O(\theta^3)\norm{u}$ because
$(\phi^\wedge)^2=O(\theta^2)$).

\subsection{Type-correct Hessian of the chart-quadratic cost}

Write $q:=\phi-\phi^*$.
Along the geodesic $t\mapsto R\exp(t\xi^\wedge)$ (constant body
velocity $\xi$), set $v(t):=J_r^{-1}(\phi(t))\,\xi$, so
$\dot\phi=v$ and $\dot f=q^\top H v$.
Differentiating once more,
\feq{
	\Hess f(R)[\xi,\xi]
	= \underbrace{\bigl(J_r^{-1}\xi\bigr)^{\!\top} H
	\bigl(J_r^{-1}\xi\bigr)}_{T_1}
	\;+\;
	\underbrace{q^\top H\,
	\Bigl(D J_r^{-1}(\phi)\bigl[J_r^{-1}(\phi)\,\xi\bigr]
	\Bigr)\xi}_{T_2},
	\label{eq:hess}
}
which is the exact second derivative of $f$ along the geodesic: the
first term is the pushforward of the chart Hessian $H$, and the
second collects the curvature of the chart through the derivative
of $J_r^{-1}$.
Formula~\eqref{eq:hess} is type-correct ($T_2$ is linear in $q$
and quadratic in $\xi$, with $D J_r^{-1}$ evaluated in the
direction of the \emph{chart} velocity $J_r^{-1}\xi$) and reduces
at $\phi=0$ to $\Hess f=H$ as it must.

\subsection{Bounding the two terms}

\paragraph{Principal term.}
Since the singular values of $J_r^{-1}$ are
$\{1,\kappa,\kappa\}$ with $\kappa\geq1$,
$\norm{J_r^{-1}\xi}\geq\norm{\xi}$, and therefore, for unit $\xi$,
\feq{
	T_1 \;\geq\; \lambda_{\min}(H)\,\norm{J_r^{-1}\xi}^2
	\;\geq\; \lambda_{\min}(H) \;=\; 0.3,
	\label{eq:T1}
}
an \emph{analytic} bound requiring no numerics
($\sin(\theta/2)\leq\theta/2$ gives $\kappa\geq1$).

\paragraph{Correction term.}
Let $u:=J_r^{-1}(\phi)\,\xi$ for unit $\xi$, so
$\norm{u}\leq\kappa(\theta)$.
Applying~\eqref{eq:DJ} and the triangle inequality,
\feq{
	\bigl\| \bigl(DJ_r^{-1}(\phi)[u]\bigr)\xi \bigr\|
	\;\leq\;
	\tfrac12\norm{u\times\xi}
	+ \beta'(\theta)\,\theta^{2}\,\kappa(\theta)
	+ 2\,\beta(\theta)\,\theta\,\kappa(\theta),
	\label{eq:threeterms}
}
using $|(\beta'/\theta)(\phi^\top u)|\leq\beta'\norm{u}$,
$\norm{(\phi^\wedge)^2\xi}\leq\theta^2$, and
$\norm{(u^\wedge\phi^\wedge+\phi^\wedge u^\wedge)\xi}
\leq2\theta\norm{u}$.
The first term admits a decisive sharpening: writing
$u=\xi+\tfrac12\phi\times\xi+\beta(\phi^\wedge)^2\xi$, the leading
contribution $\xi\times\xi=0$ vanishes, so
\feq{
	\norm{u\times\xi}
	\;\leq\; \frac{\theta}{2}+\beta(\theta)\,\theta^{2},
	\label{eq:cross}
}
an $O(\theta)$ quantity rather than the naive
$\norm{u}\approx1$; without~\eqref{eq:cross} the certificate
would fail ($\tfrac12\norm{u}\approx0.5$ already exceeds the budget
$0.3/(\bar q\lambda_{\max})\approx0.3$).
Collecting~\eqref{eq:threeterms}--\eqref{eq:cross} and
$\norm{q}\leq\bar q:=0.3+\sqrt{0.0785}$,
$\lambda_{\max}(H)=1.7$,
\feq{
	|T_2|\;\leq\;\bar q\,\lambda_{\max}(H)\,C_1(\theta),
	\qquad
	C_1(\theta):=
	\tfrac12\Bigl(\tfrac{\theta}{2}+\beta\theta^2\Bigr)
	+\beta'\theta^2\kappa
	+2\beta\theta\kappa .
	\label{eq:T2bound}
}
Every factor in $C_1$ is a power series in $\theta$ with positive
coefficients~\eqref{eq:series}, hence so is $C_1$ itself; in
particular $C_1$ is increasing on $[0,0.3]$ and its supremum is
$C_1(0.3)$.

\section{Certified evaluation}\label{sec:certified}

\begin{proposition}[Endpoint ball-arithmetic certificate]
	\label{prop:interval}
	In rigorous ball arithmetic over exact rational inputs, the
	following enclosures hold:
	\begin{center}
	\begin{tabular}{lll}
		\hline
		quantity & certified bound & method\\
		\hline
		$\kappa(0.3)$ & $\leq 1.00375987$
		& outward-rounded propagation from $\beta$\\
		$\beta(0.3)$ & $\in[0.08345860,\,0.08345861]$
		& positive-coefficient series~\eqref{eq:series} $+$ tail
		$\leq10^{-11}$\\
		$\beta'(0.3)$ & $\in[0.00083691,\,0.00083692]$
		& differentiated series $+$ tail $\leq10^{-10}$\\
		$\sup_{[0,0.3]} C_1 = C_1(0.3)$ & $\leq 0.12909469$
		& monotone $\Rightarrow$ endpoint\\
		$\bar q$ & $\leq 0.58017852$
		& $3/10+\sqrt{157/2000}$\\
		$|T_2|$ & $\leq 0.12732654$
		& product of the above\\
		$\norm{\nabla f}$ & $\leq 0.991$
		& $\lambda_{\max}\bar q\,\kappa(0.3)$\\
		\hline
		$\mu = 0.3-|T_2|$ & $\geq 0.17267346$ & \\
		\hline
	\end{tabular}
	\end{center}
	Consequently $\Hess f\succeq0.172\,I_3$ on $\ccalC$, proving
	Theorem~\ref{thm:convex}.
\end{proposition}

\begin{proof}[Method]
	The certificate is an \emph{endpoint} evaluation, and this is
	the entire point of organizing the estimate around the
	positive-coefficient series~\eqref{eq:series}: every summand of
	$C_1$ in~\eqref{eq:T2bound} is a product of such series, so
	$C_1$ is increasing on $[0,0.3]$ and
	$\sup_{[0,0.3]}C_1=C_1(0.3)$; no subdivision of the interval
	is required.
	The truncations carry explicit tail enclosures: the $\beta$
	tail after the $\theta^6$ term is bounded by
	$\sum_{k\geq5}\tfrac{|B_{2k}|}{(2k)!}\theta^{2k-2}
	\leq\sum_{k\geq5}\tfrac{4}{(2\pi)^{2k}}\theta^{2k-2}
	<10^{-11}$ for $\theta\leq0.3$ (using
	$|B_{2k}|/(2k)!=2\zeta(2k)/(2\pi)^{2k}$); the analogous
	explicit tail for $\beta'$ ($<10^{-10}$) and the algebraic
	propagation of the $\beta$ enclosure to $\kappa$
	($<10^{-12}$) are collected in Appendix~\ref{app:tails}.
	Evaluation is performed in Arb ball arithmetic
	(\texttt{python-flint} 0.9.0, 256-bit precision) with all
	inputs exact rationals, namely $3/10$, $7/10$, $23/100$, $16/100$,
	$0.0785=157/2000$, $17/10$; each row of the table is a
	machine-verified inequality between balls, e.g.\
	$C_1(0.3)\in[0.1290946799\pm7.4\times10^{-11}]$.
	Combining with the exact values $\lambda_{\min}(H)=0.3$,
	$\lambda_{\max}(H)=1.7$ and the analytic
	bound~\eqref{eq:T1} yields the certified modulus.
	Two independent, clearly \emph{nonrigorous} cross-checks
	accompany the certificate: a subdivision of $[0,0.3]$ into
	$4000$ subintervals evaluated in \texttt{mpmath}'s interval
	arithmetic, basic operations only as its documentation advises, reproduces the same suprema, and the minimum
	eigenvalue of the exact geodesic Hessian (closed-form second
	derivative along geodesics; central differences agree to
	$5\times10^{-7}$)
	(step $10^{-5}$) over
	$4000$ points sampled uniformly in the Euclidean
	principal-log coordinate ball of $\ccalC$ (fixed seed) is
	$0.254$, comfortably above the certificate
	(Fig.~\ref{fig:hess}); the gap reflects the worst-case
	alignment assumed in~\eqref{eq:T2bound}.
	The script reproducing every number, together with package
	versions and the full printed ball outputs, accompanies this
	note.
\end{proof}

\begin{figure}[t]
	\centering
	\includegraphics[width=0.62\textwidth]{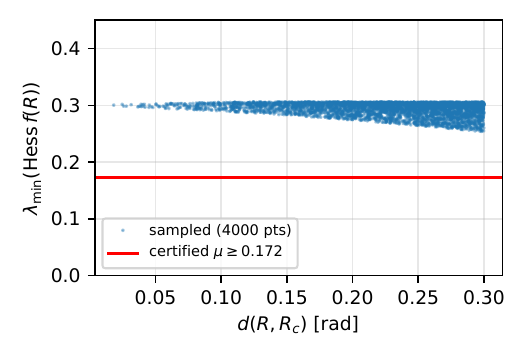}
	\caption{Sampled minimum eigenvalue of the geodesic Hessian of
		$f$ over $4000$ chart-ball--uniform points of
		$\ccalC=\overline{\ccalB}_{0.3}(I_3)$ versus distance from
		the center, together with the certified lower bound of
		Proposition~\ref{prop:interval}.
		Strong convexity holds with a wide margin everywhere on the
		ball.}
	\label{fig:hess}
\end{figure}

\section{The exact boundary violation}\label{sec:violation}

Differentiating~\eqref{eq:f} along $\dot R=R\omega^\wedge$ gives
$\dot f=q^\top H J_r^{-1}(\phi)\,\omega$, so the body-frame gradient
is
\feq{
	\nabla f(R) = J_r^{-1}(\phi)^{\!\top} H\,q,
	\qquad q=\phi-\phi^*.
	\label{eq:grad}
}
At the boundary point $R_b=\exp(0.3\,e_1^\wedge)$ we have
$\phi_b=0.3\,e_1$ and
$q_b=(0.07,\,-0.16,\,0)^\top$, hence
\nfeq{
	Hq_b=\bigl(0.07-0.7\cdot0.16,\;
	0.7\cdot0.07-0.16,\;0\bigr)^\top
	=\Bigl(-\tfrac{21}{500},\,-\tfrac{111}{1000},\,0\Bigr)^{\!\top}.
}
Because $\phi_b$ is along $e_1$, both $\phi_b^\wedge$ and
$(\phi_b^\wedge)^2$ annihilate $e_1$, so
$J_r^{-1}(\phi_b)$ (and its transpose) act as the \emph{identity}
on the $e_1$ component:
$\bigl(\nabla f(R_b)\bigr)_1=(Hq_b)_1=-21/500$ exactly, while the
transverse components mix through the invertible
$2\times2$ block.
The inward unit direction at $R_b$ is
$\hat e_c=(\log(R_b^\top I_3))^\vee/0.3=-e_1$, so
\nfeq{
	\hat e_c^\top\nabla f(R_b)
	= -\bigl(\nabla f(R_b)\bigr)_1
	= \frac{21}{500}
	\;>\;0,
}
proving Theorem~\ref{thm:violation} with no transcendental
evaluation at all.
Equivalently, with outward unit normal
$\hat n_{\mathrm{out}}:=-\hat e_c=e_1$,
$\hat n_{\mathrm{out}}^\top[-\nabla f(R_b)]=21/500$: the
\emph{descent velocity} points strictly outward while the gradient has a positive inward radial projection.
Numerically,
$\nabla f(R_b)$ equals
$(-0.0420,\,-0.1102,\,0.0167)$; the minimizer
direction $\hat e_{R^*}$ at $R_b$ makes an angle of $66.5^\circ$
with $\hat e_c$, and $d(R_b,R^*)=0.1741$: the flow is being pulled
hard toward $R^*$, but the pull has an outward radial shadow.
Figure~\ref{fig:levels} shows the geometry in the
$\phi_3=0$ slice. On the boundary sphere $\norm{\phi}=\rho$ the identity
$\phi^\top J_r^{-\top}=\phi^\top$ gives
\nfeq{
	\hat e_c^\top\nabla f
	=-\rho^{-1}\phi^\top H(\phi-\phi^*),
}
so the violating arc is exactly the set where the quadratic
form $\phi^\top H(\phi-\phi^*)$ is negative;
Fig.~\ref{fig:boundary} plots the radial gradient
component around the boundary circle in the $e_1$--$e_2$ axis
plane, exhibiting an entire arc of violation around $R_b$.

\begin{figure}[t]
	\centering
	\includegraphics[width=0.56\textwidth]{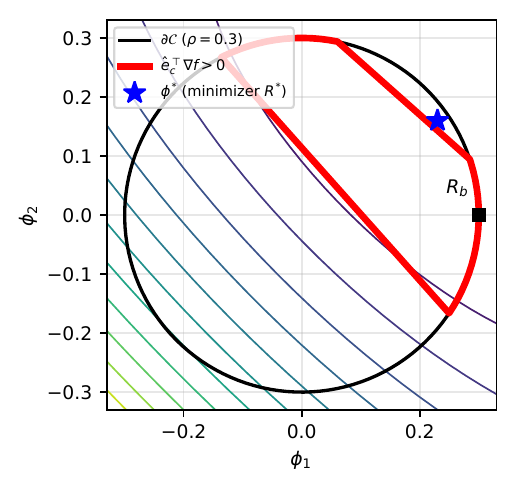}
	\caption{The $\phi_3=0$ slice of the chart: level sets of $f$, the tilted ellipses of $H$, the operating disk of radius
		$\rho=0.3$, the interior minimizer $\phi^*$, and the
		boundary arc (red) on which the inward-pointing descent
		condition fails, meaning the descent velocity points outward.
		The arrow is the \emph{chart velocity}
		$\dot\phi=-J_r^{-1}J_r^{-\top}Hq$ of the gradient flow at
		$R_b$, whose first component is exactly $+21/500$ and whose third vanishes, so the arrow is exact in the drawn plane: the
		negative Riemannian gradient has a positive projection onto
		the geodesic direction toward $R^*$, yet an outward radial
		component relative to the ball center, so the trajectory
		exits the disk.
		All data computed from~\eqref{eq:f}--\eqref{eq:grad}.}
	\label{fig:levels}
\end{figure}

\begin{figure}[t]
	\centering
	\includegraphics[width=0.6\textwidth]{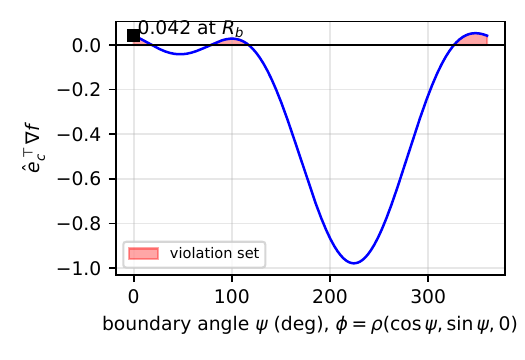}
	\caption{Radial gradient component
		$\hat e_c^\top\nabla f$ along the boundary circle
		$\phi=\rho(\cos\psi,\sin\psi,0)$.
		The violation set (shaded) is an open arc; at $\psi=0$
		(the point $R_b$) the value is exactly $21/500$.}
	\label{fig:boundary}
\end{figure}

\section{The gradient flow exits the ball}\label{sec:flow}

The violation is not a boundary-only technicality, and the exit
itself is a theorem, not a numerical observation.

\begin{corollary}[Certified exit]\label{cor:exit}
	Let $R(t)$ solve $\dot R=R(-\nabla f)^\wedge$ with
	$R(0)=R_b$.
	There exists $\varepsilon>0$ such that
	$d(R(t),R_c)>\rho$ for all $t\in(0,\varepsilon]$: the flow
	leaves $\ccalC$ immediately.
\end{corollary}
\begin{proof}
	Near $R_b$ the flow evolves in the principal chart with
	$R(t)\neq R_c$, so $t\mapsto d(R(t),R_c)$ is smooth for small
	$t$, and its derivative at $t=0$ equals
	$-\hat e_c^\top(-\nabla f(R_b))=\hat e_c^\top\nabla f(R_b)
	=21/500>0$ by Theorem~\ref{thm:violation}.
	Continuity of the derivative gives $d(R(t),R_c)>d(R_b,R_c)=\rho$
	on a right-neighborhood of $0$.
\end{proof}

The size of the excursion is then a numerical observation:
integrating with the geometric Euler scheme
$R\leftarrow R\exp(-h\nabla f^\wedge)$, the distance to the
center rises, crosses $\rho=0.3$, peaks, and only then turns
around and converges to the interior minimizer $R^*$
(Fig.~\ref{fig:flow}); the initial slope matches the exact rate
$0.042$ of Theorem~\ref{thm:violation}, and step-size refinement
shows the peak is stable to five digits:
\begin{center}
\begin{tabular}{lcc}
	\hline
	step $h$ & peak $d$ & peak time\\
	\hline
	$10^{-3}$ & $0.314291$ & $0.848$\\
	$5\times10^{-4}$ & $0.314285$ & $0.848$\\
	$2.5\times10^{-4}$ & $0.314281$ & $0.848$\\
	\hline
\end{tabular}
\end{center}
These values are \emph{not} interval-certified; the certified
statement is Corollary~\ref{cor:exit}.
For contrast, the radial cost $\tfrac12 d^2(\cdot,R^*)$ started at
the same boundary point produces a trajectory whose distance to
$R^*$ decreases monotonically and which never leaves the
ball: radial costs satisfy the inward-pointing descent
condition automatically (Section~\ref{sec:discussion}).
The numerically observed excursion of $\approx0.0143$ rad
beyond the boundary, small in absolute terms, is fatal for any
argument whose certified ingredients \emph{presuppose} membership
in $\ccalC$: logarithmic-injectivity margins, Hessian bounds, and
strong-convexity constants established from membership in
$\ccalC$ are all invalidated, even though the manifold-level
facts, with constant curvature $1/4$ and a single-valued logarithm far beyond the observed excursion, survive.

\begin{figure}[t]
	\centering
	\includegraphics[width=0.62\textwidth]{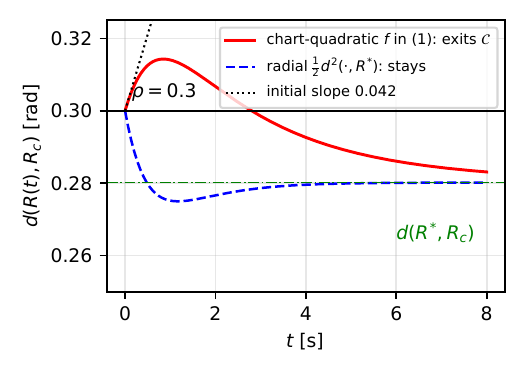}
	\caption{Distance to the center along the gradient flow of $f$
		started at the boundary point $R_b$ (red): the trajectory
		exits $\ccalC$ (certified, Corollary~\ref{cor:exit}),
		peaks at $\approx0.3143$, numerical at $h=10^{-3}$, with the step-size table for refinement, and converges to $R^*$ from outside--in.
		The radial cost from the same start (blue, dashed) never
		exits.
		Dotted: the exact initial slope $21/500$. Immediate exit is certified by Corollary~\ref{cor:exit}; the maximum excursion and subsequent convergence toward $R^*$ are numerical observations.}
	\label{fig:flow}
\end{figure}

\section{Discussion: mechanism, scope, and the sufficient repair}
\label{sec:discussion}

\paragraph{Mechanism.}
Strong convexity is a statement about the component of
$\nabla f(R)$ along the direction toward the minimizer:
$\inner{\nabla f(R)}{\hat e_{R^*}}\leq-\mu\,d(R,R^*)<0$.
It says nothing about the component along the inward radial
direction $\hat e_c$ of a ball centered \emph{elsewhere}.
At $R_b$, the inward direction toward the ball center and the
direction toward the minimizer differ by $66.5^\circ$; this
misalignment is created by the geometry of the two distinct centers
$R_c$ and $R^*$, not by the Hessian.
The off-diagonal entry $H_{12}=0.7$ then rotates the gradient
relative to the minimizer displacement $q$, enough for the
(certified, large) attraction toward $R^*$ to cast an outward
shadow of exactly $21/500$ on the radial axis.
No curvature is needed for the phenomenon (the same $H$ and
$\phi^*$ violate inward pointing for the Euclidean quadratic on a
disk), but on $\SO(3)$ the certification must additionally control
the chart curvature through
$T_2$~\eqref{eq:T2bound}, and the example shows the control succeeds
with room to spare, $0.3$ versus $0.127$.

\paragraph{Scope.}
The obstruction is single-agent: it concerns one gradient flow at
one boundary point.
As an obstruction to \emph{deriving} the boundary condition from
strong convexity it is graph-independent: any analysis that
certifies boundary behavior of local Riemannian gradients from
convexity constants alone faces it, for tree networks exactly as
for dense ones.
A consensus term can, however, oppose the outward descent
component along particular trajectories, so no claim is made that
every distributed protocol must exit.
It also calibrates solution concepts: with discontinuous
(signum-type) consensus couplings, the companion distributed
analyses establish that under the inward-pointing descent
condition the operating ball is \emph{strongly
invariant}, in the sense that every Filippov solution initialized in it remains in
it, although the pointwise tangency test fails at synchronized
boundary strata, and the boundary hypothesis itself cannot be
derived from convexity.

\paragraph{The sufficient repair.}
The condition the counterexample violates is also the condition
that suffices.
Say $f$ satisfies the \emph{inward-pointing descent condition} on
$\ccalC=\overline{\ccalB}_\rho(R_c)$ if
$\hat e_c^\top\nabla f(R)\leq0$ for every
$R\in\partial\ccalC$.
Then along $\dot R=R(-\nabla f)^\wedge$ the radial derivative at
the boundary is
$\tfrac{d}{dt}d(R,R_c)=-\hat e_c^\top\bigl(-\nabla f\bigr)
=\hat e_c^\top\nabla f\leq0$, so $\ccalC$ is invariant for the
(single-valued) gradient flow by Nagumo's theorem, and remains
strongly invariant when signum consensus couplings are added,
as the companion analyses establish; compare viability
theory~\cite{aubin2009viability}.
The condition is \emph{automatic} for the entire radial class
$f=\varphi(d(\cdot,R^*))$ with $\varphi'\geq0$, $\varphi'(0)=0$, a smooth even extension
at the origin, and $R^*\in\ccalC$: the gradient is $-\varphi'\,\hat e_{R^*}$, and
the first-variation formula together with the strong geodesic
convexity of the ball for radii below $\pi/2$ gives
$\hat e_c^\top\hat e_{R^*}\geq0$ whenever $R^*$ lies in the
ball~\cite{docarmo1992riemannian}; the quantitative halfspace
form appears in the companion manuscript.
It is thus a mild, checkable hypothesis, yet, as certified above,
a genuinely \emph{additional} one.
Companion submissions by the authors develop the distributed
finite-time consensus and optimization theory on $\SO(3)$ under
this condition; the present note is the self-contained,
quantitative reason the condition appears there as a standing
assumption.

\paragraph{Reproducibility.}
A single Python script reproduces, from scratch, the certified
table of Proposition~\ref{prop:interval} (Arb ball arithmetic via
\texttt{python-flint}~0.9.0 at 256-bit precision, exact rational
inputs, full ball outputs printed), the \texttt{mpmath} interval
cross-check, the exact violation, the step-size table, and all four
figures, with the sampling seed fixed; it is included as ancillary files
with the arXiv submission together with package versions.

\appendix
\section{Explicit Tail Enclosures}\label{app:tails}

This appendix makes the truncation certificates self-contained.
All three quantities are evaluated at the exact rational
$\theta=3/10$, and over $[0,3/10]$ via the positive-coefficient endpoint argument, so the only analytic input is a remainder
bound for each truncated series.

\paragraph{Series.}
From $x\cot x=1-\sum_{k\geq1}2^{2k}|B_{2k}|\,x^{2k}/(2k)!$ with
$x=\theta/2$ and
$1-\beta(\theta)\theta^{2}=(\theta/2)\cot(\theta/2)$,
\nfeq{
	\beta(\theta)=\sum_{k\geq1}b_k\,\theta^{2k-2},
	\quad
	\beta'(\theta)=\sum_{k\geq2}(2k-2)\,b_k\,\theta^{2k-3},
	\quad
	b_k:=\frac{|B_{2k}|}{(2k)!}=\frac{2\zeta(2k)}{(2\pi)^{2k}},
}
with every coefficient positive; $\zeta(2k)<2$ gives
$b_k<4/(2\pi)^{2k}$.

\paragraph{Tail for $\beta$ (truncation through $\theta^{6}$,
$k\leq4$).}
For $\theta\leq3/10$,
\nfeq{
	0\;\leq\;\beta(\theta)-\beta_4(\theta)
	\;=\;\sum_{k\geq5}b_k\theta^{2k-2}
	\;\leq\;\frac{4\,\theta^{8}}{(2\pi)^{10}}\,
	\frac{1}{1-(\theta/2\pi)^{2}}
	\;<\;2.8\times10^{-12}\;<\;10^{-11}.
}

\paragraph{Tail for $\beta'$ (truncation through $\theta^{5}$,
$k\leq4$).}
Consecutive tail terms shrink by
$\tfrac{2k}{2k-2}(\theta/2\pi)^{2}<3\times10^{-3}$, so the $k=5$
term dominates:
\nfeq{
	0\;\leq\;\beta'(\theta)-\beta'_4(\theta)
	\;=\;\sum_{k\geq5}(2k-2)\,b_k\theta^{2k-3}
	\;\leq\;\frac{32\,\theta^{7}}{(2\pi)^{10}}\cdot
	\frac{1}{1-3\times10^{-3}}
	\;<\;7.4\times10^{-11}\;<\;10^{-10}.
}

\paragraph{$\kappa$ as an algebraic consequence.}
$\kappa(\theta)=\sqrt{(1-\beta(\theta)\theta^{2})^{2}
+\theta^{2}/4}$ is a Lipschitz algebraic function of
$(\beta,\theta)$: on the range of interest, where
$0<1-\beta\theta^{2}\leq\kappa$,
\nfeq{
	\Bigl|\frac{\partial\kappa}{\partial\beta}\Bigr|
	=\frac{\theta^{2}\,(1-\beta\theta^{2})}{\kappa}
	\;\leq\;\theta^{2}\;\leq\;\frac{9}{100}.
}
No independent series is required: enclosing $\beta$ by its
truncation plus the tail above and evaluating the closed form in
outward-rounded ball arithmetic over the exact rational
$\theta=3/10$ encloses $\kappa$ with additional radius at most
$(9/100)\cdot10^{-11}<10^{-12}$.

\paragraph{Numerical confirmation.}
Summing the series to $k=60$ at $50$-digit precision
(\texttt{mpmath}) gives measured remainders
$\beta(3/10)-\beta_4(3/10)=1.37\times10^{-12}$ and
$\beta'(3/10)-\beta'_4(3/10)=3.66\times10^{-11}$, inside the
displayed bounds, and reproduces
$\kappa(3/10)=1.00375987$ to all shown digits; the verification
script is included with the ancillary code.

\bibliographystyle{plain}
\bibliography{references}

\end{document}